\documentclass[11pt]{article}
\usepackage{fullpage}
\usepackage{amsfonts,amsmath,amssymb,latexsym,comment}

\ifx\pdftexversion\undefined
  \usepackage[a4paper,colorlinks,linkcolor=black,filecolor=black,citecolor=black,urlcolor=black,pdfstartview=FitH]{hyperref}
\else
  \usepackage[a4paper,colorlinks,linkcolor=blue,filecolor=blue,citecolor=blue,urlcolor=blue,pdfstartview=FitH]{hyperref}
\fi
\usepackage{epsfig}
\usepackage{graphicx}
\usepackage{xspace}
\usepackage{subfigure}

\newtheorem{theorem}{Theorem}
\newtheorem{claim}{Claim}
\newtheorem{corollary}{Corollary}

\newcommand{\namedref}[2]{\hyperref[#2]{#1~\ref*{#2}}}
\newcommand{\sectionref}[1]{\namedref{Section}{#1}}

\newcommand{\figureref}[1]{\namedref{Figure}{#1}}
\newcommand{\claimref}[1]{\namedref{Claim}{#1}}
\newcommand{\lemmaref}[1]{\namedref{Lemma}{#1}}

\newcommand{\tableref}[1]{\namedref{Table}{#1}}

\newcommand{\hide}[1]{}

\newcommand{\tb}{\makebox[0.4cm]{}}
\newcommand{\due}{\makebox[0.8cm]{}}

\def\beginsmall#1{\vspace{-\parskip}\begin{#1}\itemsep-\parskip}
\def\endsmall#1{\end{#1}\vspace{-\parskip}}

\usepackage{xspace}

\newcommand\V{\mathcal{V}}

\newtheorem{remark}{{\sc Remark}\rm }[section]

\newcommand{\ignore}[1]{}
\newcommand{\byzantine}[0]{\mbox{\emph{Byzantin}\hspace{-0.15em}\emph{e}}\xspace}

\newcommand{\gradecast}{\textsc{gradecast}\xspace}

\newcounter{linenumbers}
\newcommand{\linenumberspecific}[1]{{\tt #1}}
\newcommand{\linenumber}{\stepcounter{linenumbers}\linenumberspecific{\arabic{linenumbers}:}}
\newcommand{\lineref}[1]{Line~{\tt #1}}
\newcommand{\ie}{\emph{i.e.,\ }}

\newcommand{\Ie}{\emph{I.e.,\ }}

\def\squarebox#1{\hbox to #1{\hfill\vbox to #1{\vfill}}}
\newcommand{\qed}{\hspace*{\fill}
\vbox{\hrule\hbox{\vrule\squarebox{.667em}\vrule}\hrule}\smallskip}

\newenvironment{proof}{\noindent{\bf Proof:~~}}{\(\qed\)}

\newcommand{\simpleAgree}{\textsc{ByzConsensus}\xspace}
\newcommand{\approxAlg}{\textsc{ApproxAgree}\xspace}

\newcommand{\floor}[1]{\left\lfloor#1 \right\rfloor}
\newcommand{\set}[1]{\left\{#1 \right\}}
\newcommand{\sog}[1]{\left(#1 \right)}
\newcommand{\tri}[1]{\left<#1 \right>}
\newcommand{\abs}[1]{\left|#1 \right|}
\newcommand{\avg}{\mbox{\it AVG}\xspace}

\begin{document}

\title{Simple Gradecast Based Algorithms}

\author{Michael Ben-Or\footnotemark[2] \and Danny Dolev\footnotemark[2]\and Ezra N. Hoch\footnotemark[2]}
\maketitle

\renewcommand{\thefootnote}{\fnsymbol{footnote}}

\footnotetext[2]{The Hebrew University of Jerusalem, Jerusalem, Israel}
\renewcommand{\thefootnote}{\arabic{footnote}}

\begin{abstract}
  Gradecast is a simple three-round algorithm presented by Feldman and Micali. The current work
  presents a very simple synchronous algorithm that utilized Gradecast to achieve \byzantine agreement.
  Two small variations of the presented algorithm lead to improved algorithms for solving the Approximate agreement problem and the Multi-consensus problem.

  An optimal approximate agreement algorithm was presented by Fekete, which supports up to $\frac{1}{4}n$ \byzantine
  nodes and has message complexity of $O(n^k)$, where $n$ is the number of nodes and $k$ is the number of rounds.
  Our solution to the approximate agreement problem is optimal, simple and reduces the message complexity
  to $O(k \cdot n^3)$, while supporting up to $\frac{1}{3}n$ \byzantine nodes.

  Multi consensus was first presented by Bar-Noy {\it et al}. It consists of consecutive executions of
  $\ell$ \byzantine consensuses. Bar-Noy {\it et al.,} show an optimal amortized solution to this problem,
  assuming that all nodes start each consensus instance at the same time, a property that cannot be guaranteed with early stopping.
  Our solution is simpler, preserves round complexity optimality, allows early stopping and does not require synchronized starts of the consensus instances.
\end{abstract}

\section{Introduction}
\byzantine consensus \cite{citeulike:85198} is one of the fundamental problems in the field of distributed algorithm.
Since its appearance it has been the focus of much research and many variations of the \byzantine consensus
problem have been suggested (see \cite{LynchBook, 983102}). In the current work we are interested in the \byzantine consensus problem and two such variations: multi consensus and approximate agreement.

In the \byzantine consensus problem each node $p$ has an input value $v_p$, and all non-faulty nodes
are required to reach the same output value $v$ ({\it ``agreement''}), s.t. if all non-faulty nodes have the same input value $v'$ then the output value is $v'$, \ie $v=v'$ ({\it ``validity''}).

Approximate agreement \cite{Rapprox} aims at reaching an agreement on a value from the Real domain, s.t. the output
values of non-faulty nodes are at most $\epsilon$ apart; and are within the range of non-faulty nodes' inputs. The multi consensus problem \cite{OptimalAmortized} consists of sequentially executing $\ell$ \byzantine consensuses one after the other.

The first two problems can be solved in a way that overcomes the $O(t)$ round complexity lower bound of \byzantine consensus \cite{fischer1982lbt}, where $t$ is the number of faulty nodes. Approximate agreement overcomes the lower bound by relaxing the ``validity'' property.
Regarding multi consensus, it is reasonable to think that the $O(t)$ lower bound leads to an $O(\ell \cdot t)$
lower bound for multi consensus. However, \cite{OptimalAmortized} shows how to solve $\ell$ sequential \byzantine consensuses in $O(\ell + t)$ rounds, assuming synchronized starts of the different consensuses instances.

In all three problems it is interesting to compare the round-complexity when there are $f<t$ failures. That is,
it is known that $t$ must be $< \frac{1}{3}n$ ($n$ is the number of nodes in the system). However, what if
in a specific run there are only $f < t$ failures? Can the \byzantine consensus / approximate agreement / multi consensus problem be solved quicker? The answer is ``yes'' on all three accounts.
The property of terminating in accordance to the actual number of failures $f$ is termed ``early-stopping''.
The three solutions presented in this paper all have the early-stopping property.

The solutions presented herein all use Gradecast as a building block. Gradecast was first presented in \cite{62225},
and has been used in many papers since (for example, \cite{citeulike:7123123}).
In Gradecast a single node gradecasts its value $v$, and each non-faulty node $p$ has a pair of output values:
a value $v_p$ and a confidence $c_p \in \set{0,1,2}$. The confidence $c_p$ provides information regarding the
$v_q$ values obtained at other non-faulty nodes, $q$ (see more in \sectionref{sec:gradecast}), and thus allows $p$ to
reason about $q$'s output value.

Specifically, we use Gradecast to detect faulty nodes and ignore them in future rounds.
The idea to try and identify faulty nodes and ignore them
in future rounds (not necessarily using Gradecast) has been around for some time (for example \cite{DBLP:conf/podc/AbrahamDH08, OptimalAmortized,96565,citeulike:6589507,FastAgree87,GM}). \cite{OptimalAmortized} uses it to achieve efficient sequential composition of \byzantine consensuses. In \cite{citeulike:6589507} a similar notion of ``identifying''
faulty nodes and ignoring them is used to efficiently solve the approximate agreement problem.
In essence, our usage of Gradecast ``transforms'' Byzantine failures into crash failures.

Gradecast provides a simplification of the above notion. By using Gradecast we ensure that either a \byzantine node $z$ discloses its faultiness, or all non-faulty nodes see the same message from $z$. By using a very simple
iterative algorithm we solve \byzantine consensus problem, multi consensuses and approximate agreement. All solutions are simple, optimal in their resiliency ($t < \frac{1}{3}n$), stop-early and optimal in their running time (up to a
constant factor induced by using Gradecast in each iteration). Moreover, for the approximate agreement and multi consensuses, our solutions improve upon
previously known solutions.

\subsection{Related Work}
{\bf Approximate Agreement:}
Approximate agreement was presented in \cite{Rapprox}. The synchronous solution provided in \cite{Rapprox} supports $n > 3t$ and the convergence rate is $\frac{1}{\floor{\frac{n-2t+1}{t}}+1}$ per round, which asymptotically  is $\sim \left(\frac{t}{n-2t}\right)^k$ after $k$ rounds. To easily compare the different algorithms, we consider the number of rounds it takes to reach convergence of $\frac{1}{n}$. For \cite{Rapprox}, within $O\left(\log n\right)$ rounds the algorithm ensures all non-faulty nodes have converged to $\frac{1}{n}$.
The message complexity of \cite{Rapprox} is $O(n^2)$ per each round of the $k$ rounds.

In \cite{citeulike:6589507} several results are presented.
First, for \byzantine failures there is a solution that tolerates $n > 4t$ and converges to $\frac{1}{n}$
within $O\left(\frac{\log n}{\log \log n}\right)$ rounds.
For crash failures, \cite{citeulike:6589507}
provides a solution tolerating $n > 3t$ that converges to $\frac{1}{n}$ within $O\left(\frac{\log n}{\log \log n}\right)$ rounds. The message complexity of both algorithms is $O(n^k)$. Moreover, \cite{citeulike:6589507} shows a lower bound for the \byzantine case of
$O\left(\frac{\log n}{\log \log n}\right)$ rounds to reach $\frac{1}{n}$ convergence.

Using failure-transformers, the crash resistant algorithm from \cite{citeulike:6589507} can be transformed into a \byzantine resistant algorithm (for example \cite{62588}). Such a translation has a constant multiplicative overhead in the round complexity. The transformed algorithm is tolerant to $n > 3t$ and has the original convergence rate up to a constant factor.

\cite{ZamskyPHD} solves the approximate agreement problem while tolerating $n > 3t$ \byzantine failures;
it converges to $\frac{1}{n}$ within $O\left(\frac{\log n}{\log \log n}\right)$ rounds. Moreover, \cite{ZamskyPHD}
presents algorithms with short messages for small ratios of \byzantine nodes ($\frac{n}{t} \rightarrow \infty$), but when $n > 4t$.
it requires exponential message size.

The solution presented in the current paper has a better convergence rate than that of \cite{Rapprox};
it has a higher \byzantine tolerance ratio than that of \cite{citeulike:6589507, ZamskyPHD} (\ie $n >3t$ instead of $n > 4t$) and also has an exponential improvement in the message complexity over that of \cite{ZamskyPHD} and \cite{citeulike:6589507} (from $O(n^k)$ to $O(k \cdot n^3)$). Moreover, the presented solution is simple and has a shorter presentation and much simpler proofs than the solutions of \cite{citeulike:6589507, ZamskyPHD}.

\begin{table}
\center{
  \caption{Comparison of different approximation algorithms}
  \begin{tabular}{ |p{2.15in}|| c| c| c| c|}
  \hline
& Rounds & Resiliency & Message comp. & Early-stopping?\\
\hline
\cite{Rapprox}'s approximate agreement & $O\left(\log n\right)$ & $n > 3 \cdot t$ & $O(k \cdot n^2)$ & \checkmark\\
\hline
\cite{citeulike:6589507}'s ``direct'' algorithm & $O\left(\frac{\log n}{\log \log n}\right)$ & $n > 4 \cdot t $ & $O(n^k)$ & $$\\
\hline
\cite{citeulike:6589507}'s ``indirect'' algorithm (crash-failure + transformation) & $O\left(\frac{\log n}{\log \log n}\right)$ & $n > 3 \cdot t $ & $O(n^k)$ & $$\\
\hline
\cite{ZamskyPHD}'s approximate algorithm & $O\left(\frac{\log n}{\log \log n}\right)$ & $n > 3 \cdot t $ & $O(n)$ as $\frac{n}{f}\rightarrow \infty$ & $$\\
\hline
current & $O\left(\frac{\log n}{\log \log n}\right)$ & $n > 3 \cdot t$ & $O(k \cdot n^3)$ & \checkmark\\
\hline
\end{tabular}
  \label{table:summaryApprox}
}
\end{table}
\mbox{}\\
\noindent {\bf Multi Consensus:}
The algorithm {\sc Multi-Consensus} presented in \cite{OptimalAmortized} solves $\ell$ sequential \byzantine
consensuses within $O(t + \ell)$ rounds and is resilient to $n > 3t$. However, \cite{OptimalAmortized} assumes that the starts of the different $\ell$ consensuses are synchronized, a property that cannot be ensured when a consensus stops early.
In the current paper we show how to adapt ideas from \cite{citeulike:3153966} such that our solution
does not require synchronized starts of the different consensuses.

\hide{See \tableref{table:summary} for a summary of the comparison.}

\hide{
\begin{table}
 \center{
  \caption{Comparison of current work with {\sc Multi-Consensus} \cite{OptimalAmortized}}
  \begin{tabular}{ |l|| c| c| c|}
  \hline
& Rounds & Total bits & Supports unsynchronized starts?\\
\hline
{\sc Multi-Consensus} \cite{OptimalAmortized} & $O\left(t+ \ell\right)$ & $O\left(n \cdot t^3 + \ell \cdot n \cdot t\right)$ & $$\\
\hline
current & $O\left(t+ \ell\right)$ & $O\left(n \cdot t^3 + \ell \cdot n \cdot t^2\right)$ & \checkmark\\
\hline
\end{tabular}
  \label{table:summary}
}
\end{table}
} 

In summary, a main contribution of this work is its simplicity. Using gradecast as a building block we present a very simple basic algorithm that solves the \byzantine consensus problem and two small variations of it that solve multi consensus and approximate agreement. All three algorithms support $n > 3t$, have the early-stopping property and are asymptotically optimal in their running time (up to a constant multiplicative factor). Aside from the simplicity, following are the properties of the presented algorithms:
\beginsmall{enumerate}
  \item The basic algorithm solves the \byzantine consensus problem  and terminates within\\ $3 \cdot \min \set {f+2, t+1}$ rounds.
  \item  The first variation solves the approximate agreement problem, with convergence rate of $\sog{\frac{t}{n-2t}}^k \cdot \frac{1}{k^k}$, per $3 \cdot k$ rounds (\ie within $O(\frac{\log n}{\log \log n})$ rounds it converges to $\frac{1}{n}$). The message complexity is $O(k \cdot n^3)$ per $k$ rounds, as opposed to $O(n^k)$ of the previous best known results. Moreover,
    the solution dynamically adapts to the number of failures at each round.
  \item  The second variation solves $\ell$ sequential \byzantine consensuses within $O(t + \ell)$ rounds, and  efficiently overcomes the requirement of synchronized starts
      of the consensus instances (a requirement assumed by \cite{OptimalAmortized}).
\endsmall{enumerate}

We start with \sectionref{sec:model} that presents the assumed model.
In \sectionref{sec:simple} the basic algorithm is presented and is proved to solve the \byzantine consensus problem.
The proofs are straightforward and are used as an intuitive introduction to the basic Gradecast schema.
\sectionref{sec:approx} presents a variation of the algorithm that solves the approximate agreement.
\sectionref{sec:consecutive} describes how to solve the multi consensus problem.
Lastly, \sectionref{sec:conclude} summarizes and concludes the work.

\section{Model}\label{sec:model}
The system consists of $n$ nodes, out of which up to $t < \frac{1}{3}n$ may be \byzantine, \ie behave
arbitrarily and collude together. Denote by $f \leq t$ the actual number of faulty nodes in a given run.
Communication is assumed to be synchronous and is done via message passing.
The communication graph is complete graph .

The \byzantine consensus problem consists of each node $p$ having an input value $v_p$ from a finite set $\V$ (\ie $v_p \in \V$).
Each node $p$ also has an output value $o_p \in \V$. Two properties should hold:
\beginsmall{enumerate}
  \item {\it ``agreement''}: $o_p=o_q$ for any two non-faulty nodes $p,q$ (thus we can talk about the output value of the algorithm);
  \item {\it ``validity''}: if all non-faulty nodes start with the same input value $v$, then the output value of the algorithm is $v$.
\endsmall{enumerate}


\begin{figure*}[t]\center

\begin{minipage}{5.0in}
\hrule \hrule \vspace{1.7mm} \footnotesize
\setlength{\baselineskip}{3.9mm} \noindent Algorithm \gradecast$(q,\mbox{{\sc ignore}}_p)$
 \vspace{1mm} \hrule \hrule
\vspace{1mm}

\begin{tabular}{ r l }

& \textit{/* Initialization */}\hfill\textit{/* executed on node $p$ with leader node $q$*/}\\
\linenumber & {\bf set} ignore all messages being received below from nodes in $\mbox{\sc ignore}_p$;\\
\linenumber & {\bf if} $p=q$ {\bf then} $v=$ `the input value';\\
\\
& \textit{/* Dissemination */}\\
\linenumber & {\bf round 1} The leader $q$ sends $v$ to all:\\
\linenumber & {\bf round 2} $p$ sends the value received from $q$ to all;\\
\\
& \tb \textit{/* Notations */}\\
\linenumber & \tb {\bf let} $\tri{j,v_j}$ represent that $p$ received $v_j$ from $j$;\\
\linenumber & \tb {\bf let} $maj$ be a value received the most among such values;\\
\linenumber & \tb {\bf let} $\#maj$ be the number of occurrences of $maj$;\\
\\
& \textit{/* Support */}\\
\linenumber & {\bf round 3} {\bf if} $\#maj\ge n-t$ {\bf then} $p$ sends $maj$ to all;\\
\\
& \tb \textit{/* Notations */}\\
\linenumber & \tb {\bf let} $\tri{j,v'_j}$ represent that $p$ received $v'_j$ from $j$;\\
\linenumber & \tb {\bf let} $maj'$ be a value received the most among such values;\\
\linenumber & \tb {\bf let} $\#maj'$ be the number of occurrences of $maj'$;\\
\\
& \tb \textit{/* Grading */}\\
\linenumber & \tb {\bf if} $\#maj'\ge n-t$ {\bf set} $v_p := maj'$ and $c_p := 2$;\\
\linenumber & \tb {\bf otherwise, if} $\#maj'\ge t+1$ {\bf set} $v_p := maj'$ and $c_p := 1$;\\
\linenumber & \tb {\bf otherwise} {\bf set} $v_p=\perp$ and $c_p=0$;\\
\\
\linenumber & {\bf return} $\tri{q,v_p,c_p}$;
\end{tabular}

\vspace{1mm}
\hrule \hrule
\end{minipage}

 \caption{\gradecast: The Gradecast protocol}\label{figure:gradecast}
\end{figure*}

\subsection{Gradecast}\label{sec:gradecast}
Gradecast \cite{62225} is a distributed algorithm that ensures some properties that are similar to those of broadcast.
Specifically, in Gradecast there is a sender node $p$ that sends a value $v$ to all other nodes. Each node $q$'s
output is a pair $\tri{v_q, c_q}$ where $v_q$ is the value $q$ thinks $p$ has sent and $c_q$ is $q$'s confidence
in this value. The Gradecast properties ensure that:
\beginsmall{enumerate}
  \item if $p$ is non-faulty then $v_q=v$ and $c_q=2$, for every non-faulty $q$;
  \item for every non-faulty nodes $q, q'$: if $c_q > 0$ and $c_{q'} > 0$ then $v_q=v_{q'}$;
  \item $|c_q-c_{q'}| \leq 1$ for every non-faulty nodes $q, q'$.
\endsmall{enumerate}

The protocol in \figureref{figure:gradecast} is basically the original protocol presented in~\cite{62225} with explicit handling of boycotting messages coming from nodes known to be faulty.

\begin{theorem}\label{thm:gradecast}
  There is a 3 round Gradecast algorithm.
\end{theorem}
\begin{proof}
 \figureref{figure:gradecast} presents such an algorithm.  Assuming that at initiation {\sc ignore}$_p$, for every non-faulty node $p$, contains only faulty nodes, then the proof of~\cite{62225} holds.
\end{proof}

 The implementation in \figureref{figure:gradecast}  implies the following claim.

\begin{claim}\label{claim:ignore}
If the leader $q$, $q\in \mbox {\sc ignore}_p$ for every non-faulty $p$  then following the completion of \gradecast, $c_p=0$ for every non-faulty $p$.
\end{claim}
\begin{proof}
Every non-faulty node ignores all messages send by $q$ and as a result every non-faulty node will return $c_p=0$.
\end{proof}

\section{Simple \byzantine Consensus}\label{sec:simple}
The idea behind the \simpleAgree algorithm (\figureref{figure:simpleByzAgree}) is to use gradecast as a means of forcing the \byzantine nodes to ``lie''
at the expense of being expelled from the algorithm. That is, at each iteration a node $p$ will gradecast its own value, and then consider the values it received: a) any node that gradecasted a value with confidence $ \leq 1$ will be marked as faulty, and will be ignored for the rest of the algorithm; b) any value with confidence $\geq 1$ will be
considered, and $p$ will update its own value to be the majority of values with confidence $\geq 1$. Moreover, this mechanism ensures that for a faulty node $z$, if different non-faulty nodes consider different values for $z$'s gradecast, then at least one of them should obtain the value with zero confidence.  For example, one considers $z$ gradecasted ``0'' with confidence 1, and the other considers $z$ gradecast's confidence to be 0. The result of such a case is that all non-faulty nodes will mark $z$ to be faulty, and will remove it from the algorithm.
In other words, a \byzantine node can produce, using \gradecast, contradicting values to non-faulty nodes at most once.

\begin{figure*}[t,h!]\center
\setcounter{linenumbers}{0}
\begin{minipage}{5.0in}
\hrule \hrule \vspace{1.7mm} \footnotesize
\setlength{\baselineskip}{3.9mm} \noindent Algorithm $\simpleAgree$
 \vspace{1mm} \hrule \hrule
\vspace{1mm}

\begin{tabular}{ r l }

& \textit{/* Initialization */}\hfill\textit{/* executed on node $p$*/}\\
\linenumber & {\bf set} $BAD := \emptyset$;\\
\\
& \textit{/* Main loop */}\\
\linenumber & {\bf for} $r := 1$ to $t+1$ do:\\
\linenumber & \tb {\bf \gradecast}$(p,BAD)$ with input value $v$;\\
\\
& \tb \textit{/* Notations */}\\
\linenumber & \tb {\bf let} $\tri{q,v,c}$ represent that $q$ gradecasted $v$ with confidence $c$;\\
\linenumber & \tb {\bf let} $maj$ be the value received the most among values with confidence $\geq 1$;\\
& \due (if there is more than one such value, take the lowest)\\
\linenumber & \tb {\bf let} $\#maj$ be the number of occurrences of $maj$ with confidence $2$;\\

\\
& \tb \textit{/* Updates */}\\
\linenumber & \tb {\bf set} $v := maj$;\\
\linenumber & \tb {\bf set} $BAD := BAD \cup \set{q ~|~ \mbox{received $\tri{q,*,c}$ with $c \leq 1$}}$;\\
\linenumber & \tb {\bf if} $\#maj \geq n-t$ then break loop;\\
\linenumber & {\bf end for}\\
\\
\linenumber & {\bf if} executed for $< t+1$ iterations then participate in one more iteration;\\
\linenumber & {\bf return} $v$;
\end{tabular}

\vspace{1mm}
\hrule \hrule
\end{minipage}

 \caption{\simpleAgree: a simple \byzantine consensus algorithm}\label{figure:simpleByzAgree}
\end{figure*}


Denote by $\cup BAD_r$ the union of all $BAD$ variables for non-faulty nodes at the beginning of iteration $r$.
Similarly, denote by $\cap BAD_r$ the intersection of all $BAD$ variables for non-faulty nodes at the beginning of iteration $r$.

\begin{claim}\label{claim:gradecast}
  If $\cup BAD_r$ contains only faulty nodes,
  then the properties of gradecast, following its execution in \lineref{3}, hold.
\end{claim}
\begin{proof}
  Ignoring messages of faulty nodes does not affect the properties of gradecast, since gradecast
  works properly no matter what the faulty nodes do. Specifically, ignoring messages from faulty nodes
  is equivalent to the faulty nodes not sending those messages.
\end{proof}

\begin{claim}\label{claim:BAD}
  If $\cup BAD_r$ contains only faulty nodes, then $\cup BAD_{r+1}$ contains only faulty nodes.
\end{claim}
\begin{proof}
  Consider a non-faulty node $q$. By \claimref{claim:gradecast}, $q$'s gradecast confidence is 2 at all
  non-faulty nodes. Thus, no non-faulty node adds $q$ to $BAD$ in the current iteration. Therefore,
  $\cup BAD_{r+1}$ contains only faulty nodes.
\end{proof}

\begin{corollary}
  The \gradecast invoked in \lineref{3} satisfies the gradecast properties, and
  $\cup BAD_r$ never contains non-faulty nodes.
\end{corollary}
\begin{proof}
  By iteratively applying \claimref{claim:gradecast} and \claimref{claim:BAD}.
\end{proof}

\begin{claim}\label{claim:samev}
  If at the beginning of some iteration all non-faulty nodes have the same value $v$, then all
  non-faulty nodes that are still in the main loop exit the loop and update their value to $v$.
\end{claim}
\begin{proof}
  All non-faulty nodes see at least $n-f$ copies of $v$ with confidence 2. Thus, by \lineref{5,7} they all
  update their value to $v$, and (if they are still in the loop) by \lineref{9} they all exit it;
\end{proof}

\begin{claim}\label{claim:BADinc}
  If non-faulty nodes $p,q$ have different values of $maj$ at iteration $r$, then $|\cap BAD_{r+1}| > |\cap BAD_r|$.
\end{claim}
\begin{proof}
  If $p$ has a different value of $maj$ than $q$, then (w.l.o.g.) by the definition of $maj$ (\lineref{5}) there is some \byzantine node $z$ such that $p$ received $\tri{z,u,*}$ from $z$'s gradecast, and $q$ received
  $\tri{z,u',*}$, s.t. $u \neq u'$. By the properties of gradecast, all non-faulty nodes have confidence of at most 1 for $z$'s gradecast. Therefore, by \lineref{8}, all
  non-faulty nodes add $z$ to $BAD$. That is, $z \in \cap BAD_{r+1}$.

  To conclude the proof, we need to show that $z \notin \cap BAD_r$. Since $p$ and $q$ see different confidence for $z$'s gradecast, we conclude that some non-faulty node didn't ignore $z$'s messages. (Otherwise, by \claimref{claim:ignore}, $z$ gradecast confidence would have been 0 at all non-faulty nodes.) Therefore, we conclude that $z \notin \cap BAD_r$.
\end{proof}

\begin{claim}\label{claim:sameMaj}
  If all non-faulty nodes have the same value of $maj$ at iteration $r$, then all non-faulty nodes end
  iteration $r$ with the same value $v$.
\end{claim}
\begin{proof}
  Immediate from \lineref{7}.
\end{proof}

\begin{claim}\label{claim:breakloop}
  If some node $p$ breaks the main loop due to \lineref{9} during iteration $r$, then all non-faulty nodes end
  iteration $r$ with the same value $v$.
\end{claim}
\begin{proof}
  For $p$ to pass the condition of \lineref{9}, $\#maj$ must be at least $n-t$. That is, $p$ sees at least $n-t$
  gradecast values equal to $maj$ with confidence 2. From the properties of gradecast, all other non-faulty nodes
  see $n-t$ gradecast values equal to $maj$ with confidence $\geq 1$. By \lineref{5,7} they all update their
  value to be that same value.
\end{proof}

\begin{theorem}
  \simpleAgree solves the \byzantine consensus problem.
\end{theorem}
\begin{proof}
  From \claimref{claim:samev} it is clear that {\it ``validity''} holds. To show that {\it ``agreement''} holds we consider
  two different cases. First, if a non-faulty node passes the condition of \lineref{9} in the first $t$ iterations, then by \claimref{claim:breakloop}
  and \claimref{claim:samev} {\it ``agreement''} holds.

  Second, if no non-faulty node ever passes the condition of \lineref{9} in the first $t$ iterations, then all non-faulty nodes perform the main loop of \simpleAgree $t+1$ times.
  By \claimref{claim:sameMaj} and \claimref{claim:samev} this means that in every iteration of the first $f$ iterations there is some pair of non-faulty nodes that have different values of $maj$.
  By \claimref{claim:BADinc} $\abs{\cap BAD_{t+1}} > \abs{\cap BAD_{t}} > \dots > \abs{\cap BAD_{1}}=0$. Thus, $\abs{\cap BAD_{t+1}} \geq t$.
  Therefore, in iteration $t+1$ all non-faulty nodes ignore all \byzantine nodes' messages. Therefore, all non-faulty nodes see the same set of gradecasted messages (all with confidence 2) and thus they all agree on the value of $maj$. By \lemmaref{claim:sameMaj} all non-faulty nodes end iteration $t+1$ with the same value of $v$.
\end{proof}

\begin{remark}\label{rem:earlyStopping}
  Notice that the above proof also proves the ``early stopping'' property of \simpleAgree. More specifically,
  if there are $f \leq t$ actual failures, then \simpleAgree terminates within $\min \set {f+2,t+1}$ iterations (each iteration takes 3 rounds).
\end{remark}

\section{Approximate Agreement}\label{sec:approx}
In this section we are interested in an algorithm that solves the approximate agreement problem~\cite{Rapprox}. Approximate agreement is somewhat different from \byzantine agreement. Specifically, each node $p$ has a real input value $v_p \in \Re$ and
a real output value $o_p \in \Re$. Denote by $L$ ($H$ resp.) the lowest (highest resp.) input values of non-faulty nodes. Given a constant $\epsilon$ the approximate agreement problem requires that:
\beginsmall{enumerate}
  \item {\it ``agreement''}: $|o_p-o_q| \leq \epsilon$ for any two non-faulty nodes $p,q$;
  \item {\it ``validity''}: $o_p \in [L, H]$ for every non-faulty node $p$.
\endsmall{enumerate}

The algorithm \approxAlg in \figureref{figure:ApproxAgree} has the following iterative structure: a) gradecast $v$ to everyone; b) collect all values received into a multi-set; c) perform some averaging method (denote it by \avg) on the multi-set, and use that as the input of the next iteration.
\avg removes the $t$ lower and higher values, then computes the average of the remaining set.

For $\epsilon = \frac{H-L}{n}$ the algorithm in \figureref{figure:ApproxAgree} requires  $O(\frac{\log n}{{\log \log n}})$ iterations. The best previous approximate agreement that has an early-stopping property, polynomial message size and supports $n > 3t$ \byzantine nodes (see \cite{Rapprox}), requires $O(\log n)$ iterations.

\begin{figure*}[t]\center
\setcounter{linenumbers}{0}
\begin{minipage}{5.0in}
\hrule \hrule \vspace{1.7mm} \footnotesize
\setlength{\baselineskip}{3.9mm} \noindent Algorithm $\approxAlg(\epsilon)$
 \vspace{1mm} \hrule \hrule
\vspace{1mm}

\begin{tabular}{ r l }

& \textit{/* Initialization */}\hfill\textit{/* executed on node $p$*/}\\
\linenumber & {\bf set} $BAD := \emptyset$;\\
\\
& \textit{/* Main loop */}\\
\linenumber & {\bf while} $true$ do: \\
\linenumber & \tb {\bf \gradecast}$(p,BAD)$ with input value $v$;\\
\\
& \tb \textit{/* Notations */}\\
\linenumber & \tb {\bf let} $\tri{q,v,c}$ represent that $q$ gradecasted $v$ with confidence $c$;\\
\linenumber & \tb {\bf let} $values$ be the multiset of received values with confidence $\geq 1$,\\
& \due \ {\bf and} add ``0'' until $values$ contains $n$ items;\\
\linenumber & \tb {\bf let} $values'$ be the multiset of received values with confidence $2$;\\

\\
& \tb \textit{/* Updates */}\\
\linenumber & \tb {\bf set} $v := \avg(values)$;\\
\linenumber & \tb {\bf set} $BAD := BAD \cup \set{q ~|~ \mbox{received $\tri{q,*,c}$ with $c \leq 1$}}$;\\
\linenumber & \tb {\bf if} there are $n-t$ items in $values'$ that are at most $\epsilon$ apart, then break loop;\\
\linenumber & {\bf end while}\\
\\
\linenumber & {\bf participate} in one more iteration;\\
\linenumber & {\bf return} $v$.
\end{tabular}

\vspace{1mm}
\hrule \hrule
\end{minipage}

 \caption{\approxAlg: an efficient approximate agreement algorithm}\label{figure:ApproxAgree}
\end{figure*}

In a similar manner to \simpleAgree (\sectionref{sec:simple}) only \byzantine nodes can
be added to the $BAD$ set of any non-faulty node, and a given \byzantine node $z$'s value can be viewed
differently by different non-faulty nodes at most once. That is, each \byzantine node can ``lie''
at most once.

Denote by $V_r$ the multi-set containing the values $v$ of all non-faulty nodes at the beginning
of iteration $r$. Denote by $L(M)$ the lowest value in $M$ and by $H(M)$ the highest value in $M$.
$M^t$ is the multi-set $M$ after the $t$ lowest and $t$ highest values have been removed. Using these
notations, the \avg method is defined as: $$\avg(M) := \frac{\sum_{e{\tiny \in M^t}}e}{|M^t|}\;.$$

\begin{remark}
  The \avg method has the following properties:
  \beginsmall{enumerate}
    \item $\avg(M) \in [L(M^t), H(M^t)]$;
    \item if $M$ contains $n-t$ values in the range $[v, v+\epsilon]$ then $\avg(M) \in [v, v+\epsilon]$;
    \item for $x \leq t$: if $M$ is a multi-set of $n-x$ values, and $M_1$ and $M_2$ contain $M$ and additional $x$ items (\ie $M_1, M_2$ differ by at most $x$ values) then $$|\avg(M_1)-\avg(M_2)| \leq  \sog{H(M)- L(M)} \cdot \frac{x}{n-2t}\;.$$
  \endsmall{enumerate}
\end{remark}

\begin{proof}
  \beginsmall{enumerate}
    \item $\avg(M)$ is the average of the set $M^t$, which is clearly between the lowest value in $M^t$, and the highest value in $M^t$.
    \item Since $M$ contains $n-t$ values in the range $[v, v+\epsilon]$, then by removing the $t$ highest values we remain with values that are all at most $v+\epsilon$; that is, $H(M^t) \leq v+\epsilon$. Similarly $L(M^t) \geq v$. Thus, the average of $M^t$ is in the range $[v, v+\epsilon]$.
    \item Since $|M_1^t|=|M_2^t|=n-2t$, we need to evaluate the difference between $|\sum_{e\in M_1^t}e-\sum_{e\in M_2^t}e|$. Since $M_1$ and $M_2$ differ by at most $x$ values, $M_1^t, M_2^t$ differ by at most $x$
        values as well. Since $x \leq t$, each of these values is in the range $[L(M), H(M)]$;
        therefore, $|\sum_{e\in M_1^t}e-\sum_{e\in M_2^t}e| \leq \sog{H(M)-L(M)} \cdot x$. By dividing both sides by $|M^t_1|$ we have that $|\avg(M_1)-\avg(M_2)| \leq  \sog{H(M)- L(M)} \cdot \frac{x}{n-2t}$.
  \endsmall{enumerate}
\end{proof}

\begin{remark}
  Notice that for any pair of non-faulty nodes $p,q$ it holds that the set `$values$' of $p$ contains at least $n-f$
  exact same values as `$values$' of $q$. Moreover, `$values'$' of $p$ is contained in `$values$' of $q$.
\end{remark}

\begin{claim}\label{claim:aproxValid}
  For non-faulty $p$, at the end of iteration $r$, the value $v$ is in the range $\left[L(V_{r-1}),
  H(V_{r-1})\right]$.
\end{claim}
\begin{proof}
  Immediate from the first property of \avg, and the fact that all non-faulty nodes' values are
  in the set $values$ of $p$ (which stems from Gradecast's properties).
\end{proof}

\begin{claim}\label{claim:Exit1}
  If a non-faulty node $p$ exits the main loop in \lineref{9} in iteration $r$ then
  $H(V_r) - L(V_r) \leq \epsilon$.
\end{claim}
\begin{proof}
    Due to the properties of gradecast, $values$ of node $q$ contains the set $values'$ of node $p$.
  Thus, if $p$ passes the condition in \lineref{9} then $values$ of $q$ contains $n-t$ values in the range $[v, v+\epsilon]$ (for some $v$). Therefore, by the second property of \avg, node $q$'s computed value
  will be in the range $[v, v+\epsilon]$. This claim holds for every non-faulty $q$, for the same value of $v$.
   That is, all non-faulty nodes compute their new value to be in the range $[v, v+\epsilon]$; \Ie within $\epsilon$ of each other.
\end{proof}

\begin{claim}\label{claim:Exit2}
  If $H(V_r) - L(V_r) \leq \epsilon$ for some iteration $r$, then every non-faulty node $p$ that is still in
  the main loop, exits
  the main loop (\lineref{9}) during iteration $r$.
\end{claim}
\begin{proof}
  In iteration $r$ every non-faulty node $p$ sees $n-f$ values with confidence 2 that are within $\epsilon$ of each other and thus (if $p$ is still in the main loop) passes the condition of \lineref{9}.
\end{proof}

Denote by $NEW_r := |\cap BAD_{r+1}|-|\cap BAD_{r}|$; \ie $NEW_r$ is the number of \byzantine nodes detected
as faulty (by all non-faulty nodes) during iteration $r$.

\begin{claim}\label{claim:singleIter}
  For every iteration $r$ it holds that  $H(V_{r+1})-L(V_{r+1}) \leq \left(H(V_{r})-L(V_{r})\right) \cdot \frac{NEW_r}{n-2t}$.
\end{claim}
\begin{proof}
  We consider two cases. First, if $NEW_r = 0$ then no new \byzantine node is added to $\cap BAD_{r+1}$. \Ie,
  every \byzantine node $z \notin \cap BAD_{r}$ is seen by some non-faulty node $p$ as having gradecast value
  with confidence 2. Thus, every non-faulty node sees the same gradecast value of $z$. For \byzantine $z \in \cap BAD_{r}$ all non-faulty nodes ignore $z$'s messages.
  Therefore, all non-faulty nodes have the same value of the set $values$ and thus they all update $v$ in the same manner. \Ie $H(V_{r+1})-L(V_{r+1})=0$.

  Continue with the case where $NEW_r > 0$. In a similar manner to the first case, for every node $z \in \cap BAD_{r}$ and for every node $z \notin \cap BAD_{r+1}$ the non-faulty nodes have the same gradecast value of $z$.
  Thus, for two non-faulty nodes $p,q$ there are at most $|\cap BAD_{r+1}|-|\cap BAD_{r}| = NEW_r$ different values
  in their $values$ sets. By the third property of the $\avg$ method, $\avg(values)$ of $p$ and $\avg(values)$ of $q$ are at most $\sog{H(V_{r})-L(V_{r})}\cdot \frac{NEW_r}{n-2t}$ apart.
\end{proof}

\begin{claim}\label{claim:kk}
  Assume $\approxAlg(\epsilon)$ runs for $k$ iterations, and no non-faulty node has exited the main loop. Then, $H(V_{k+1}) - L(V_{k+1}) \leq \frac{H-L}{k^k} \cdot \sog{\frac{t}{n-2t}}^k$.
\end{claim}

\begin{proof}
  At each iteration $r$, by \claimref{claim:singleIter} $H(V_{r+1})-L(V_{r+1}) \leq \left(H(V_{r})-L(V_{r})\right) \cdot \frac{NEW_{r}}{n-2t}$. That is, after the $k$'s iteration we have
  that $H(V_{k+1})-L(V_{k+1}) \leq \left(H-L\right) \cdot \prod_{i=1}^k \frac{NEW_i}{n-2t}$. The worst
  value of $\prod_{i=1}^k \frac{NEW_i}{n-2t}$ is reached when $NEW_i=NEW_j$ for every $i,j \in \set{1,\dots,k}$.
  \Ie when $NEW_i = \frac{t}{k}$.

  From the above we have that $H(V_{k+1})-L(V_{k+1}) \leq \left(H-L\right) \cdot \prod_{i=1}^k \frac{\frac{t}{k}}{n-2t} = \left(H-L\right) \cdot \sog{\frac{t}{n-2t}}^k \cdot \frac{1}{k^k}$.
\end{proof}

\begin{theorem}
  $\approxAlg(\epsilon)$ solves the approximate agreement problem, and for $\epsilon = \frac{H-L}{n}$ $\approxAlg(\epsilon)$ converges within at most $O\left(\frac{\log n} {\log \log n}\right)$ rounds.
\end{theorem}
\begin{proof}
  {\it ``validity''} holds by iteratively applying \claimref{claim:aproxValid}.

  If some node terminates, then by \claimref{claim:Exit1} and \claimref{claim:Exit2} all nodes terminate within one iteration.
  Moreover, by \claimref{claim:Exit2} and by applying \claimref{claim:kk} for large enough values of $k$
  we have that eventually some node terminates; thus proving {\it ``agreement''}.

  Assume towards contradiction that $\approxAlg(\frac{H-L}{n})$ runs for $r=\lceil \frac{\log n}{\log \log n} \rceil$ iterations and has not terminated yet. By \claimref{claim:kk} we have that
   $H(V_{r+1})-L(V_{r+1}) \leq \frac{H-L}{r^r}$.
   Consider $\log \left(\frac{\log n}{{\log \log n}} ^ {\frac{\log n}{{\log \log n}}}\right) =
   \frac{\log n}{{\log \log n}} \cdot \log (\frac{\log n}{{\log \log n}}) =
   \frac{\log n}{{\log \log n}} \cdot (\log \log n - \log \log \log n)$. Since $\log \log n - \log \log \log n \geq \frac{1}{2} \log \log n$ (for large values of $n$) we have that $\log \left(\frac{\log n}{{\log \log n}} ^ {\frac{\log n}{{\log \log n}}}\right) \geq \frac{1}{2} \log n$. Concluding that $\frac{\log n}{{\log \log n}} ^ {\frac{\log n}{{\log \log n}}} \geq \sqrt{n}$. Therefore, $H(V_{r+1})-L(V_{r+1}) \leq \frac{H-L}{\sqrt{n}}$. By running the algorithm for $2r = 2 \lceil \frac{\log n}{\log \log n} \rceil$ we have that $H(V_{2r+1})-L(V_{2r+1}) \leq \frac{H-L}{n} = \epsilon$.

   By \claimref{claim:Exit2} every non-faulty node terminates by the end of iteration $2r+2$. Each iteration
   is composed of a constant number of rounds, hence $\approxAlg\left(\frac{H-L}{n}\right)$ terminates within $O\left (\frac{\log n}{\log \log n} \right)$ rounds.
\end{proof}

\begin{remark}
Notice that \approxAlg has the early stopping property in two senses. First, if there are $f \leq t$ failures then
the convergence rate for $k$ iterations is $\sog{\frac{f}{n-2t}}^k \cdot \frac{1}{k^k}$. Second,
in each iteration, the convergence of \approxAlg depends solely on the number of discovered failures; which leads
to a quick termination of the algorithm if in some iteration no new failures are discovered (or if few new failures
are discovered).
\end{remark}

\section{Sequential Executions of Consensuses}\label{sec:consecutive}
In \cite{OptimalAmortized} the authors investigate the Multi-consensus problem: suppose we want to execute a sequence of
\byzantine consensuses, such that each consensus possibly depends on the output of the previous consensus. That is,
we must execute the consensuses sequentially and not concurrently. (The concurrent version is termed ``interactive consistency'', initially stated in \cite{Agree80}.)

The first solution that comes to mind is to simply execute $\ell$ instances of \byzantine consensus one after the other. However, due to the $O(f)$ round lower bound on \byzantine consensus \cite{fischer1982lbt}, a naive sequential
execution will lead to a running time of $O(\ell \cdot f)$. In \cite{OptimalAmortized} they give
a solution that has total running time of $O(f + \ell)$.

However, \cite{OptimalAmortized} assumes that at each consensus instance all correct nodes start the consensus at once; that is, they are always synchronized. This implies that
for every one of the $\ell$
consensuses, all nodes know exactly when the consensus starts. This assumption is problematic, since due to the early-stopping nature of the algorithm used in \cite{OptimalAmortized}, the different nodes may
terminate each invocation of a consensus at different rounds, leading to a problem with the synchronized starts assumption.
Assuming synchronized starts, the algorithm \simpleAgree can easily (almost without any change) produce the same result; \Ie for $\ell$ consecutive consensuses, \simpleAgree's running time will be $O(f + \ell)$.

In this section we consider two results. First, we analyze \simpleAgree assuming synchronized starts of each consensus, and compare it to the results of \cite{OptimalAmortized}. Second, we augment \simpleAgree with
ideas from \cite{citeulike:3153966} such that \simpleAgree can solve $\ell$ consensuses within $O(f+\ell)$ rounds
even if the initiations of each consensus are not synchronized among the non-faulty nodes.

\subsection{Synchronized Starts}
We start with the assumption that all $\ell$ consensuses have synchronized starts. That is, all non-faulty nodes
start the $i$th (out of $\ell$) instance of \byzantine consensus in the same round. With this assumption in place,
\simpleAgree can be used almost as-is. The single modification required is to perform the initialization state (\lineref{1}) only once for the entire sequence of $\ell$ consensuses (and not once per consensus).

It is easy to see that the following two statements hold also for the modified algorithm:
\beginsmall{enumerate}
  \item Each \byzantine node $z$ can cause non-faulty nodes to disagree on $z$'s gradecasted value at most once throughout the sequence of $\ell$ consensuses;
  \item The number of iterations of a given consensus instance $i$ is $\min \{f+2, t+1\}$, where $f$ is the actual number of \byzantine nodes for which there are non-faulty nodes that do not agree on their gradecasted values during instance $i$.
\endsmall{enumerate}

Using the above statements it is easy to conclude that for a sequence of $\ell$ sequential consensuses, the number of iterations of the modified \simpleAgree is at most $t + 2 \cdot \ell$, since each consensus takes at least 2 iterations to complete, and at most $t$ more iterations are required (depending on the \byzantine nodes' behavior).

Since each iteration is composed of executing gradecast, the total round complexity is $3 \cdot t + 6 \cdot \ell$.
Notice that in each iteration there are at most $n$ gradecasts, each requiring $O(n \cdot t)$ messages. Therefore, the total bit complexity of $\ell$ consensuses is $O\left((t+\ell) \cdot n^2 \cdot t\right) = O\left(n^2 \cdot t^2 + \ell \cdot n^2 \cdot t\right)$. By a simple optimization which defines a subset of $3\cdot t + 1$ as performing the gradecast (and the rest of the nodes just ``listen'' to messages), the total bit complexity is reduced to $O\left(n \cdot t^3 + \ell \cdot n \cdot t^2\right)$.

Note that \simpleAgree is optimal with respect to its \byzantine resiliency and the total running time (\cite{OptimalAmortized} is also optimal with respect to amortized message size and total bits).
Lastly, {\sc Multi-Consensus} is not a complicated algorithm, but \simpleAgree is even simpler.

\subsection{Unsynchronized Starts}
In this subsection we remove the assumption of synchronized starts of the different $\ell$ consensuses. This
is done by implementing ideas of \cite{citeulike:3153966} in a way that is consistent with sequential executions
of \simpleAgree. In \cite{citeulike:3153966} the authors show how to sequentially execute consensuses even if the
consensuses are started at different times and terminate at different times. The reason we cannot use \cite{citeulike:3153966}'s results as-is in our case, is because \cite{citeulike:3153966} assumes that the running time of the
algorithm does not change over time; which does not hold for our solution. Specifically, \simpleAgree's $k$th execution's running time depends on what
occurred in the $k-1$ instances that preceded it.

The required addition to \simpleAgree is as follows:
\beginsmall{enumerate}
  \item If $p$ wants to terminate according to \simpleAgree, it first sends {\it ``done''} to every one, and waits;
  \item In addition, if node $p$ receives $t+1$ {\it ``done''} messages, then $p$ sends all nodes a {\it ``done''} message as well;
  \item Lastly, if node $p$ received $2 \cdot t+1$ {\it ``done''} messages, then $p$ completes the current instance of  \simpleAgree;
\endsmall{enumerate}

Notice that if $i$ is the first round in which some non-faulty node $p$ halts, then all other non-faulty nodes either halt in round $i$, or in round $i+1$. That is, non-faulty nodes terminate within one round of each other (no matter what difference there was between their starting times).

Another addition required is to increase the length of each of \simpleAgree's rounds according to the difference
between different non-faulty nodes' starting times. If the different starting times of each consensus instances is at most $\Delta$ rounds at different nodes, then each iteration of \simpleAgree needs to be increased by a factor of $\Delta$ so that messages of the $i$-th round are received (and considered as messages of the $i$-th round) by all other nodes.

\begin{remark}
  The increased length of each iteration is only with respect to the original \simpleAgree. The additions regarding the ``done'' messages are left as is. Thus, we ensure that non-faulty nodes still terminate within one round of each other.
\end{remark}

To see that the addition does not harm the correctness of \simpleAgree, consider the following. Suppose $p$ starts
running at some round, and $q$ starts running $\Delta$ rounds afterwards. By the correctness proof of \simpleAgree, if $p$ terminates, then one iteration afterwards $q$ will also terminate. But actually, there is an even stronger property: if $p$ terminates then all other non-faulty nodes already have the same value as $p$ (see \claimref{claim:breakloop}). \Ie if $q$ terminates when $p$ terminates, then it will terminate with the same value.

To conclude, notice that the above addition ensures that a non-faulty node halts only if some non-faulty node has
terminated in the execution of \simpleAgree. Thus, from that round onward, the sequence of non-faulty nodes' termination is valid.

The above discussion leads to the following theorem:

\begin{theorem}
  The updated \simpleAgree algorithm solves $\ell$ sequential \byzantine consensus within $O(\Delta \cdot (\ell + t))$ rounds, while not requiring synchronized starts of the different consensus instances.
\end{theorem}

\section{Conclusion}\label{sec:conclude}
We presented a simple algorithm \simpleAgree that uses Gradecast as a building block, and solves the \byzantine
consensus problem within $3 \cdot \min \set {f+2, t+1}$ rounds.

Two variations of \simpleAgree are given. The first one optimally solves the approximate agreement problem and reduces
the message complexity of the best known optimal algorithm from $O(n^k)$ to $O(k \cdot n^3)$.
The second variant of \simpleAgree optimally solves the multi consensus problem and efficiently supports
unsynchronized starts of the consensus instances.

All three algorithms have optimal resiliency, optimal running time (up to a constant multiplicative factor) and have the early stopping property. Aside from their improved simplicity, the two variants also improve (in different aspects) upon previously best known solutions.

\section*{Acknowledgements}
{Michael Ben-Or is the incumbent of the Jean and Helena Alfassa Chair in Computer Science, and he was supported in part by the Israeli Science Foundation (ISF) research grant. Danny Dolev is Incumbent of the Berthold Badler Chair in Computer Science. Danny Dolev was supported in part by the Israeli Science Foundation (ISF) Grant number 0397373. }

\thispagestyle{empty}
\bibliographystyle{plain}
\bibliography{bibliography}

\end{document}